\documentclass[10pt,reqno]{amsart}
\usepackage{floatflt}
\usepackage{xypic}
\usepackage{rotating}
\usepackage{bbm}
\usepackage{mathrsfs}
\usepackage{diagbox}
\usepackage{dutchcal}
\usepackage{extarrows}
\usepackage{cite}
\usepackage{amsfonts} 
\usepackage[dvipsnames,usenames]{color}
\usepackage{graphicx,latexsym,bm,amsmath,amssymb,verbatim,multicol,lscape}
\makeatletter
\@namedef{subjclassname@2020}{\textup{2020} Mathematics Subject Classification}
\makeatother
\newtheorem{thm}{Theorem} [section]

\newtheorem{lem}{Lemma}

\newtheorem{prob}{Problem}

\numberwithin{equation}{section}
\numberwithin{lem}{section}
\numberwithin{con}{section}
\numberwithin{cor}{section}
\numberwithin{rem}{section}
\numberwithin{prob}{section}

\def\bew{\begin{widetext}}
\def\eew{\end{widetext}}
\def\be{\begin{equation}}
\def\ee{\end{equation}}
\def\bea{\begin{eqnarray}}
\def\eea{\end{eqnarray}}

\allowdisplaybreaks
\begin{document}
\title[]{On Almost MDS codes holding $(n-k-1)$-designs and an open problem on Witt design}
\author[H.K. Wei]{Huakai Wei}
\address{Mathematical College, Sichuan University, Chengdu 610064, P.R. China}
\email{1810097@mail.nankai.edu.cn}
\author[S.Y. Qiang]{Shiyuan Qiang}
\address{Mathematical College, Sichuan University, Chengdu 610064, P.R. China}
\email{syqiang23@163.com}
\author[S.F. Hong]{Shaofang Hong$^*$}
\address{Mathematical College, Sichuan University, Chengdu 610064, P.R. China}
\email{sfhong@scu.edu.cn}
\begin{abstract}
Let $\mathcal{C}$ be an almost maximum distance separable (AMDS) code
over the finite fields with parameters $[n,k,n-k]$ whose minimum-weight
codewords support an $(n-k-1)$-$(n,n-k,\lambda)$ design. We first determine
the parameters of the dual code $\mathcal{C}^{\bot}$. We also characterize
the design and its complementary design supported by the minimum weight
codewords of $\mathcal{C}^{\bot}$. In particular, we present explicit formulas
for the weight enumerators of all subset codes of $\mathcal{C}^{\bot}$ with
length at least $k+1$. These formulas are completely determined by the cardinality
of the subset, the dimension, together with the parameter $\lambda$, and are
independent of the specific choice of subsets.
As an application, we answer an open problem posed by Wang, Tang and Ding in 2023
by proving that the $3$-design supported by the minimum-weight codewords in their
proposed cyclic dual code is isomorphic to the complementary design of the Witt
spherical geometry design.
\end{abstract} 

\thanks{$^*$S.F. Hong is the corresponding author and was supported partially
by National Science Foundation of China \#12571007.}
\maketitle

\section{Introduction}
Let $\mathbb{F}_q$ be the finite
field of $q$ elements, where $q=p^{m}$ with $p$ a prime
and $m$ a positive integer. Let
$\mathbb{F}_q^{\ast}:=\mathbb{F}_q\setminus\{0\}$
be the multiplicative group. Let $n\ge 1$ be an integer
and $\langle n\rangle:=\{1,2,\cdots,n\}$.
For any finite set $A$, let $\left|A\right|$ denote the cardinality of $A$.
Let ${\bf a}=(a_{1}, \cdots, a_{n})$ and
${\bf b}=(b_{1}, \cdots, b_{n})$ be two vectors over $\mathbb{F}_q$.
The {\it Hamming distance} between ${\bf a}$ and ${\bf b}$ is defined as
$d({\bf a}, {\bf b}):=\left|\{i\in\langle n\rangle \mid a_{i}\neq b_{i}\}\right|$.
The {\it support set} ${\rm supp}({\bf a})$ of ${\bf a}
\in\mathbb{F}_q^{n}$ is defined by
${\rm supp}({\bf a}):=\{i\in\langle n\rangle\mid a_{i}\neq0\}$.
The {\it Hamming weight} $w({\bf a})$ of ${\bf a}$ is defined as
$w({\bf a}):=d({\bf a}, {\bf 0})$. Then $w({\bf a})=\left|{\rm supp}({\bf a})\right|$.
The {\it minimum Hamming distance}, denoted by
$d=d(\mathcal{C})$, of a code $\mathcal{C}\subseteq
\mathbb{F}_q^{n}$ is defined to be the smallest Hamming
distance between distinct codewords.
An $[n,k,d]$ linear code $\mathcal{C}$ over $\mathbb{F}_q$ is
a $k$-dimensional subspace of the $n$-dimensional affine space $\mathbb{F}_q^n$
with minimum Hamming distance $d$.
The {\it dual code} $\mathcal{C}^{\bot}$ of the $[n,k]$ linear code $\mathcal{C}$ over $\mathbb{F}_q$ is
defined by $\mathcal{C}^{\bot}:=\{{\bf c}^{\bot}\in\mathbb{F}_q^n\mid
\langle {\bf c},{\bf c}^{\bot}\rangle=0,\forall{\bf c}\in \mathcal{C}\}$,
where $\langle {\bf c},{\bf c}^{\bot}\rangle$ denote the Euclidean
inner product of ${\bf c}$ and ${\bf c}^{\bot}$ over $\mathbb{F}_q$.
For a linear code $\mathcal{C}$, associated with any given subset
$I\subseteq\langle n\rangle$, we define the {\it subset code}
$\mathcal{C}(I)$ of $\mathcal{C}$ by
$\mathcal{C}(I):=\{{\bf c}\in\mathcal{C}\mid
{\rm supp}({\bf c})\subseteq I\}$ (see \cite{[QWH-FFA25]}).
For any integer $i$ with $0\leq i\leq n$, let $A_{i}$
stand for the number of codewords in $\mathcal{C}$
with length $n$ and the Hamming weight $i$. Then the
{\it weight enumerator} of $\mathcal{C}$ is defined
to be the polynomial $1+A_{1}x+A_{2}x^{2}+\cdots+A_{n}x^{n}.$
The sequence $(1, A_{1}, A_{2},\cdots, A_{n})$ is called
the {\it weight distribution} of $\mathcal{C}$.
For the sake of simplicity, we assume throughout this paper
that all linear codes are non-degenerate.

Several well-known AMDS codes are known to support $(n-k-1)$-designs, among which the binary
Hamming code with parameters $[7,4,3]$ is a classical example. In recent years, such codes
have attracted considerable attention. In 2021, Ding, Tang and Tonchev \cite{[DTT-DCC21]}
presented an infinite family of AMDS codes with parameters $[q+1,q-3,4]$ over $\mathbb{F}_{2^m}$
that support a $3$-$(q+1,4,2)$ design. They also determined the parameters of the dual codes,
namely $[q+1,4,q-4]$, but did not derive the general weight enumerators of these codes. In 2022,
Xu, Cao and Qu \cite{[XCQ-IT22]} presented three infinite families of AMDS codes with parameters
$[q+1,q-3,4]$, $[q,q-3,3]$ and $[q,q-4,4]$, supporting $3$-$(q+1,4,p^l-2)$, $2$-$(q,3,p^l-2)$ and
$3$-$(q,4,1)$ designs, respectively. They further determined the corresponding weight enumerators.
Subsequently, Xiang, Tang and Liu \cite{[XTL-DCC22]} constructed another infinite family of AMDS codes
with parameters $[q+1,q-3,4]$ over $\mathbb{F}_{7^m}$ that support a $3$-$(q+1,4,5)$ design, and
established the weight enumerators of their dual codes. In 2023, Wang, Tang and Ding \cite{[WTD-IT23]}
provided an infinite family of AMDS codes with parameters $[q+1,q-3,4]$ supporting a
$3$-$(q+1,4,p^l-2)$ design. They completely determined the weight enumerators of both the codes
and their duals. In 2026, Xu, Wu, Lu and Cao \cite{[XWLC-FFA26]} further generalized these results.

In this paper, we study the almost maximum distance separable (AMDS) code $\mathcal{C}$ over the
finite field $\mathbb{F}_q$ with parameters $[n,k,n-k]$ whose minimum-weight codewords support
an $(n-k-1)$-$(n,n-k,\lambda)$ design. Particularly, we subsume the above-mentioned
case-by-case computations of the weight enumerators of the dual code $\mathcal{C}^{\bot}$
into a single unified framework. We will show that if an AMDS code with parameters $[n,k,n-k]$
has its minimum weight codewords supporting an $(n-k-1)$-$(n,n-k,\lambda)$ design, then both
the parameters and the supported design of its dual code can be completely determined. Furthermore,
we derive explicit formulas for the weight enumerators of all subset codes of the dual code with
lengths at least $k+1$. These weight enumerators depend only on the subset length, the dimension
together with the parameter $\lambda$, and are independent of the specific choice of subsets.
That is, we have the first main result of this paper as follows.
\begin{thm}\label{thm1.1}
Let $\mathcal{C}$ be an AMDS code over $\mathbb{F}_q$ with parameters $[n,k,n-k]$.
If the minimum weight codewords in $\mathcal{C}$ support an
$(n-k-1)$-$(n,n-k,\lambda)$ design $\mathbb{D}=(\mathbb{P},\mathbb{B})$, then
the dual code $\mathcal{C}^{\bot}$ has parameters $[n,n-k,k+1-\lambda]$,
and the minimum weight codewords of $\mathcal{C}^{\bot}$ support an
$(n-k-1)$-$(n,k+1-\lambda,\frac{\binom{k+1}{n-k-1+\lambda}}
{\binom{k+1}{\lambda}})$ design $\mathbb{D}^{\bot}=(\mathbb{P},\mathbb{B}^{\bot})$.
The complementary design of $\mathbb{D}^{\bot}$, denoted by
${\mathbb{D}^{\bot}}'=(\mathbb{P},{\mathbb{B}^{\bot}}')$,
is an $(n-k-1)$-$(n,n-k-1+\lambda,1)$ design, where
${\mathbb{B}^{\bot}}'=\{\cup_{B\in\mathbb{B}\atop I\subseteq B}B\mid
I\subseteq\mathbb{P},\left|I\right|=n-k-1\}$.

Moreover, for any subset $S\subseteq\langle n\rangle$
with $|S|=k+r,r\in\langle n-k\rangle$,
the weight enumerator $A_{\mathcal{C}^{\bot}(S)}(z)$ of
$\mathcal{C}^{\bot}(S)$ is given by
\begin{align*}
&A_{\mathcal{C}^{\bot}(S)}(z)
=1+(q-1)\prod_{j=1}^{r-1}\frac{k+1+j}{\lambda+j}z^{k+1-\lambda}
+\sum_{i=k+2}^{k+r}z^{i}\binom{k+r}{i}\\
&\Big((-1)^{i-k}\binom{i-1}{k}+\sum_{j=0}^{i-k-1}\binom{i}{j}(-1)^{j}q^{i-k-j}
+(q-1)\binom{i}{k+1}\sum_{j=k+1-\lambda}^{k}\frac{(-1)^{i-j}
\binom{\lambda}{k+1-j}}{\binom{i-j}{k+1-j}}\Big).
\end{align*}
\end{thm}

Witt spherical geometry design is a special class of Steiner systems $\mathrm{S}(3,q+1,  q^l+1)$
admitting the projective linear group $\mathrm{PGL}(2,q)$ as an automorphism group
\cite{[W-AMSUH37]}. Owing to their rich algebraic and combinatorial properties, they have
attracted considerable attention in coding theory. In particular, Ding and Tang \cite{[DTT-DCC21]},
\cite{[DT-IT20]}, \cite{[TD-IT21]} constructed several infinite families of linear
codes related to the Witt spherical geometry design in 2020 and 2021. Building upon these
developments, Liu, Ding, Mesnager, Tang and Tonchev \cite{[LDMTT-IT22]} established an infinite family
of BCH codes over finite fields whose minimum weight codewords support the Witt spherical geometry
design. In 2023, Wang, Tang, and Ding \cite{[WTD-IT23]} investigated an infinite family of cyclic codes
$\mathcal{C}_{(q,q+1,3,\frac{p^s-1}{2})}$ whose minimum weight codewords support $3$-designs, together
with their dual codes, where $q=p^m$ with $p$ a prime, $m\geq 2$ an integer, and $s$ an integer satisfying
$1\leq s\leq m-1$. Particularly, they determined the parameters of these codes and established
their weight enumerators. Moreover, they proposed the following open problem.

\begin{prob}\label{prob1.1} (Open problem 18 in \cite{[WTD-IT23]})
Is the $3$-design supported by the minimum weight codewords of the dual code
$\mathcal{C}_{(q,q+1,3,\frac{p^s-1}{2})}^{\bot}$ isomorphic to the complementary design of
the Witt spherical geometry design with parameters $3$-$ (q+1,p^{\gcd(m,s)}+1,1)$?
\end{prob}

The second goal of this paper is to investigate the parameters and the structures supported by their
minimum weight codewords of both linear code $\mathcal{C}_{(q,q+1,3,\frac{p^s-1}{2})}$ and its dual code
$\mathcal{C}_{(q,q+1,3,\frac{p^s-1}{2})}^{\bot}$. As an application of Theorem \ref{thm1.1}, we will use
it to study Problem \ref{prob1.1}. The second and third main results of this paper can be stated as follows.

\begin{thm}\label{thm1.2}
Let $q=p^m$ with $p$ being a prime number and $m\ge 2$ being a positive integer,
and let $\mathcal{C}_{(q, q+1, 3, h)}$ denote the BCH code over $\mathbb{F}_q$.
Suppose that $2h+1\equiv p^{m'}\pmod{q+1}$, where $m'$ is a nonnegative integer.
Let $t=\gcd(m,m')$. Then each of the following is true:

{\rm (i).} If $m'\equiv 0\pmod{m}$, then $\mathcal{C}_{(q, q+1, 3, h)}$
is an MDS code with parameters $[q+1, q-2, 4]$.

{\rm (ii).} If $m'\not\equiv 0\pmod{m}$ and $p^t=2$, then $\mathcal{C}_{(q, q+1, 3, h)}$
is an MDS code with parameters $[q+1 , q-3 , 5]$.

{\rm (iii).} If $m'\not\equiv 0\pmod{m}$ and $p^t>2$, then $\mathcal{C}_{(q, q+1, 3, h)}$
is an AMDS code with parameters $[q+1 , q-3 , 4]$,
and the minimum weight codewords in $\mathcal{C}_{(q, q+1, 3, h)}$
support a $3$-$(q+1,4,p^t-2)$ design, denoted by $\mathbb{D}=(\mathbb{P},\mathbb{B})$
with $\mathbb{P}=U_{q+1}$ and
\begin{align*}
\mathbb{B}=&\{\{\mathrm{\pi}_{a,b,c,d}(0,\infty,1,\xi)\mid
\xi\in\mathbb{F}_{p^t}^{\ast}\setminus\{1\},
\{x,y,z\}\subseteq U_{q+1},x\ne y\ne z\ne x\},
\end{align*}
where $a = y(x-z)$, $b = x-z$, $c = x(z-y)$, $d = z-y$,
$\mathrm{\pi}_{a,b,c,d}(\xi) = \frac{a\xi + c}{b\xi + d}$
is a projective bijection between $\mathbb{F}_q\cup\{\infty\}$ and $U_{q+1}$.
\end{thm}

\begin{thm}\label{thm1.3}
Let $q=p^m$ with $p$ being a prime number and $m\ge 2$ being a positive integer.
Let $\mathcal{C}_{(q, q+1, 3, h)}^{\bot}$  be the dual code
of the BCH code $\mathcal{C}_{(q, q+1, 3, h)}$. Suppose that
$2h+1\equiv p^{m'}\pmod{q+1}$, where $m'$ is a nonnegative integer,
and let $t=\gcd(m,m')$.
Then each of the following is true:

{\rm (i).} If $m'\equiv 0\pmod{m}$, then $\mathcal{C}_{(q, q+1, 3, h)}^{\bot}$ is an MDS code with parameters $[q+1,3,q-1]$.

{\rm (ii).} If $m'\not\equiv 0\pmod{m}$ and $p^t=2$, then $\mathcal{C}_{(q, q+1, 3, h)}^{\bot}$ is an MDS code with parameters $[q+1,4,q-2]$,

{\rm (iii).} If $m'\not\equiv 0\pmod{m}$ and $p^t>2$, then $\mathcal{C}_{(q, q+1, 3, h)}^{\bot}$ is a linear code with parameters $[q+1,4,q-p^t ]$,
and the minimum weight codewords in $\mathcal{C}_{(q, q+1, 3, h)}^{\bot}$ support a $3$-$(q+1,q-p^t,\lambda^{\bot})$ design,
which is isomorphic to the complementary design of the Witt spherical geometry design with parameters $3$-$(q+1,p^t+1,1)$,
where $\lambda^{\bot}=\frac{\binom{q-2}{p^t+1}}{\binom{q-2}{p^t-2}}$.
Moreover, the weight enumerator of $\mathcal{C}_{(q, q+1, 3, h)}^{\bot}(S)$ for any subset $S$ with $|S|=(q-3)+r\ge q-2$ is
\begin{align*}
&A_{\mathcal{C}_{(q, q+1, 3, h)}^{\bot}(S)}(z)
=1+(q-1)\prod_{j=1}^{r-1}\frac{q-2+j}{\lambda+j}z^{q-2-\lambda}
+\sum_{i=q-1}^{q-3+r}z^{i}\binom{q-3+r}{i}\\
&\Big((-1)^{i-q+3}\binom{i-1}{q-3}+\sum_{j=0}^{i-q+2}\binom{i}{j}(-1)^{j}q^{i-q+3-j}
+(q-1)\binom{i}{q-2}\sum_{j=q-p^t}^{q-3}\frac{(-1)^{i-j}\binom{p^t-2}{q-2-j}}
{\binom{i-j}{q-2-j}}\Big).
\end{align*}

In particular, the weight enumerator of $\mathcal{C}_{(q, q+1, 3, h)}^{\bot}$ is obtained by taking $S=\langle n\rangle$.
\end{thm}

Evidently, Theorem \ref{thm1.3} gives an affirmative answer to Problem \ref{prob1.1}.
Furthermore, Theorems \ref{thm1.2} and \ref{thm1.3} extend and refine the results given by Wang,
Tang and Ding \cite{[WTD-IT23]} as well as those gotten by Xu, Wu, Lu and Cao \cite{[XWLC-FFA26]}.

This paper is organized as follows. In Section 2, we recall and show some preliminary lemmas
which are needed in the proof of Theorem \ref{thm1.1}, \ref{thm1.2} and \ref{thm1.3}. In Section 3,
we give the proof of Theorem \ref{thm1.1}. Then in Section 3, we present the details  of the proofs
of Theorems \ref{thm1.2} and \ref{thm1.3}. The final section is devoted to some interesting remarks.

\section{Preliminary lemmas}

In this section, we first supply some preliminary lemmas which are needed in the
proof of Theorem \ref{thm1.1}. We begin with the following remarkable result.

\begin{lem}\label{cite_complementary_design}\cite{[HP-CUP03]}
Given a $t$-$(v,k,\lambda)$ design $\mathbb{D}=(\mathbb{P},\mathbb{B})$,
let $\mathbb{B}^c$ be the set of the complements of the blocks
in $\mathbb{B}$. Then the complementary design of $\mathbb{D}$
is a $t$-$(v,v-k,\lambda_{(0,t)})$ design,
where $\lambda_{(0,t)}:=\lambda\frac{\binom{v-t}{k}}{\binom{v-t}{k-t}}$.
\end{lem}

\begin{lem} {\rm (M{\"o}bius inversion formula)}
\label{Mobius inversion formula}\cite{[RPS-CUP97]}
Let $\mathrm{B}_n$ be the Boolean algebra of rank $n$.
For any $T\subseteq S$, the  M{\"o}bius function is given by $$\mu(T,S)=(-1)^{|S-T|}.$$
The M{\"o}bius inversion formula for $\mathrm{B}_n$
becomes the following statement: Let $f,g:\mathrm{B}_n\to K$.
Then it holds that
$$g(S)=\sum_{T\subseteq S}f(T)$$
if and only if it holds that
$$f(S)=\sum_{T\subseteq S}(-1)^{|S-T|}g(T).$$
\end{lem}

The following lemma is equivalent to Theorem 1.4.13 in \cite{[HP-CUP03]}.
For ease of reference, we restate it here.

\begin{lem}\label{lem1}\cite{[HP-CUP03]}
Let $\mathcal{C}$ be a linear code over $\mathbb{F}_q$
with parity-check matrix $H$. For any subset
$I\subseteq\langle n\rangle$ and for any index
$j\in\langle n\rangle\setminus I$,
there exists a codeword $\mathbf{c}\in\mathcal{C}$
such that $\mathrm{supp}(\mathbf{c})\subseteq I\cup\{j\}$
and $\mathbf{c}_j\ne 0$ if and only if
$H_j\in \mathrm{span}(H_i,i\in I)$, where
$H_i$ denotes the $i$-th column of $H$.
\end{lem}

\begin{lem}\label{lem_structure}
Let $\mathcal{C}$ be an $[n, k, d]$ linear code over $\mathbb{F}_q$ and let
$\mathcal{C}^{\bot}$ be its dual code. Let $H$ be a parity-check matrix of
$\mathcal{C}$, and denote its column vectors by $H_1, \cdots, H_n \in \mathbb{F}_q^{n-k}$.
For any subset $I\subseteq \langle n\rangle$, define the subspace
$V_0:=\mathrm{span}(\{H_i \mid i \in I\})$, and let $r:=n-k-\dim V_0$ with
$0 \le r \le n - k$. Then the following statements hold:

{\rm (i).} There exist $r$ linearly independent column vectors
$H_{i_1}, \cdots, H_{i_r}$ ($1 \le i_1 < \dots < i_r \le n$) such that
$H_{i_j} \notin \mathrm{span}(\{V_0, H_{i_1}, \cdots, H_{i_{j-1}}\})$ for all $j \in \langle r\rangle.$

{\rm (ii).} For each nonnegative integer $j \le r$, let
$V_j := \mathrm{span}(\{V_0, H_{i_1}, \cdots, H_{i_j}\})$ and
$I_j := \{ l \in \langle n\rangle \mid H_l \in V_j \setminus V_{j-1} \}$,
where $V_{-1} := \emptyset$ and $I_0 := \{ l \in \langle n\rangle \mid H_l \in V_0 \}$.
Then we have the disjoint union: $\bigcup_{i=0}^r I_i = \langle n\rangle$.

{\rm (iii).} We have
$\mathcal{C}^{\bot}(\langle n\rangle \setminus I) = \mathrm{span}(\{\mathbf{c}_1^{\bot}, \cdots, \mathbf{c}_r^{\bot}\})$
with $\mathbf{c}_1^{\bot}, \cdots, \mathbf{c}_r^{\bot} \in \mathcal{C}^{\bot}$, and for each $j \in \langle r\rangle$,
$I_j \subseteq \mathrm{supp}(\mathbf{c}_j^{\bot}) \subseteq \bigcup_{i=j}^r I_i$.
\end{lem}

\begin{proof}

(i). It follows directly from the basis extension theorem in linear algebra.

(ii). It follows directly from the construction of $V_j$ and $I_j$.

(iii). Since $H$ is a parity-check matrix of the linear code $\mathcal{C}$,
we find that $\dim(\mathcal{C}^{\bot}(\langle n\rangle\setminus I))
=\dim(V_0^{\bot})=r$.

Claim I: For any codeword $\mathbf{c}^{\bot}\in\mathcal{C}^{\bot}(\langle n\rangle\setminus I)$
and index $i\in I_0$, $\mathbf{c}^{\bot}_i=0$.

Since $\mathbf{c}^{\bot}_i=0$ for any index $i\in I$, we only need to treat the indexes in $I_0\setminus I$.
For any index $i\in I_0\setminus I$, we can deduce that $H_j,j\in I$ and $H_i$ are linearly dependent.
Hence, there exists a codeword $\mathbf{c}\in \mathcal{C}$ with $\mathbf{c}_i\ne 0$ such that
$\mathrm{supp}(\mathbf{c})\subseteq I\cup\{i\}$. By the definition of the dual code, we obtain that
$$0=\langle \mathbf{c},\mathbf{c}^{\bot}\rangle
=\sum_{j\in\langle n\rangle}\mathbf{c}_j\mathbf{c}^{\bot}_j
=\mathbf{c}_i\mathbf{c}^{\bot}_i.$$
Consequently, $\mathbf{c}^{\bot}_i=0$ as $\mathbf{c}_i\ne 0$.
Therefore, by the arbitrariness, we have $\mathbf{c}^{\bot}_i=0$
for any index $i\in I_0\setminus I$. So Claim I is proved.

Claim II: For any codeword $$\mathbf{c}^{\bot}\in\mathcal{C}^{\bot}\Big(\bigcup_{i=j}^{r}I_i\Big)
=\mathcal{C}^{\bot}\Big(\langle n\rangle\setminus
\bigcup_{i=0}^{j-1}I_i\Big), j\in\langle r\rangle,$$
the components $\{ \mathbf{c}^{\bot}_i \mid i \in I_j \}$ are either
all zero or all nonzero.

If $\left|I_j\right|=1$, Claim II
is trivial. So we consider when $\left|I_j\right|>1$.
Suppose there exists an index $j_1\in I_j$ such that
$\mathbf{c}^{\bot}_{j_1}\ne 0$. Now for any index $j_2\in I_j\setminus\{j_1\}$,
applying Lemma \ref{lem1} to $H_{j_2}\in\mathrm{span}(\{H_{j_1}\}\cup\{H_i,i\in \bigcup_{i=0}^{j-1}I_i\})$
and $H_{j_1}\notin\mathrm{span}(\{H_i,i\in\bigcup_{i=0}^{j-1}I_i\})$
gives us that there is a nonzero codeword $\mathbf{c}\in \mathcal{C}(\bigcup_{i=0}^{j-1}I_i\cup\{j_1,j_2\})$
such that $\mathbf{c}_{j_1}\ne 0$ and $\mathbf{c}_{j_2}\ne 0$.
It follows that $$0=\langle \mathbf{c},\mathbf{c}^{\bot}\rangle
=\sum_{j\in\langle n\rangle}\mathbf{c}_j\mathbf{c}^{\bot}_j
=\mathbf{c}_{j_1}\mathbf{c}^{\bot}_{j_1}+\mathbf{c}_{j_2}
\mathbf{c}^{\bot}_{j_2},$$
which yields that
$$\mathbf{c}^{\bot}_{j_2}=-\frac{\mathbf{c}_{j_1}\mathbf{c}^{\bot}_{j_1}}{\mathbf{c}_{j_2}}\ne 0$$
as $\mathbf{c}_{j_1},\mathbf{c}_{j_2}$ and $\mathbf{c}^{\bot}_{j_1}$ are all nonzero.
Therefore, by the arbitrariness, we have $\mathbf{c}^{\bot}_i\ne 0$
for any index $i\in I_j$. Thus Claim II is proved.

Finally, we complete the proof of Lemma \ref{lem_structure}.
For any index $j\in\langle r\rangle$, consider the map
$$\phi:\mathcal{C}^{\bot}\Big(\bigcup_{i=j}^{r}I_i\Big)\to
\mathbb{F}_q^{|I_j|},\mathbf{c}^{\bot}\mapsto\mathbf{c}^{\bot}\mid_{I_j},$$
where $\mathbf{c}^{\bot}\mid_{I_j}$ denotes the restriction of the dual codeword
$\mathbf{c}^{\bot}$ to the coordinate set $I_j$. By Claim II, $\mathrm{Im}(\phi)$
must be a one-dimensional linear space over $\mathbb{F}_q$. This means that
$$\left| \frac{\mathcal{C}^{\bot}(\bigcup_{i=j}^{r}I_i)}
{\ker(\phi)}\right|=\left|\mathrm{Im}(\phi)\right|=q.$$
We then derive that
$$\dim\Big(\mathcal{C}^{\bot}\Big(\bigcup_{i=j+1}^{r}I_i\Big)\Big)=\dim(\ker(\phi))
=\dim\Big(\mathcal{C}^{\bot}\Big(\bigcup_{i=j}^{r}I_i\Big)\Big)-1,$$
and there exists a codeword $\mathbf{c}^{\bot}_j$ such that
$$I_j\subseteq\mathrm{supp}(\mathbf{c}^{\bot}_j)\subseteq\bigcup_{i=j}^{r}I_i.$$
From the linear independence of $\mathbf{c}^{\bot}_j,j\in\langle r\rangle$,
one can deduce that $$\mathcal{C}^{\bot}(\langle n\rangle\setminus I)
=\mathrm{span}(\{\mathbf{c}^{\bot}_1,\cdots,\mathbf{c}^{\bot}_r\}).$$
So part (iii) is proved.

This completes the proof of Lemma \ref{lem_structure}.
\end{proof}

In the rest of this section, we present several lemmas that are needed
in the proofs of Theorems \ref{thm1.2} and \ref{thm1.3}.

\begin{lem} \label{singleton} \cite{[RS-IT64]}
{\rm (Singleton bound)} For a linear $[n, k, d]$ code
$\mathcal{C}$ over $\mathbb{F}_q$, we have
$d\leq n-k+1.$
\end{lem}

\begin{lem} \label{lem_MDS}\cite{[RMR-CUP06]}
Let $\mathcal{C}$ be a linear code over $\mathbb{F}_q$.
Then $\mathcal{C}$ is an MDS code if and only if the dual
code $\mathcal{C}^{\bot}$ of $\mathcal{C}$ is an MDS code.
\end{lem}

The following Steiner systems were constructed by Witt \cite{[W-AMSUH37]},
and are called {\it Witt spherical geometry designs}.
But for ease of reference, we adopt the version stated in \cite{[LDMTT-IT22]} here.

\begin{lem}\cite{[LDMTT-IT22]}\cite{[W-AMSUH37]}
Let $\mathrm{PGL}(2,q)$ be the projective linear group of degree two.
Put $\mathbb{B}:=\mathrm{PG}(1,q)$ and $\mathrm{Orb}_{\mathbb{B}}
:=\{ g\mathbb{B}\mid g\in \mathrm{PGL}(2,q^m)\}$. Let
$\mathbb{D}:=(\mathbb{B},\mathrm{Orb}_{\mathbb{B}})$.
Then $\mathbb{D}$ is a Steiner system $\mathrm{S}(3,q+1,q^m+1)$.
\end{lem}

\begin{lem}\label{U_q+1_PG(1,q)}\cite{[LDMTT-IT22]}
The map $\mu:\mathbb{F}_q\cup\{\infty\}\to U_{q+1},\mu(x)=\frac{u_0x+1}{x+u_0}$ is a bijection,
where $u_0\in U_{q+1}\setminus\{\pm 1\}$.
\end{lem}

\begin{lem}\label{PGL2}\cite{[DT-WCP22]}
The projective general linear group $\mathrm{PGL}(2,q)$ consists of
all the following permutations of the set $\{\infty\}\cup\mathbb{F}_q$:
$$\mathrm{\pi}_{(a,b,c,d)}(x)=\frac{ax+c}{bx+d}$$ with $ad-bc\ne 0$,
and the conventions above about $\infty$. The corresponding group
action is triply transitive.
\end{lem}

\begin{lem}\label{XWLC-IT25}\cite{[XWLC-IT25]}
The minimum distance $d$ of the BCH code $\mathcal{C}_{(q, q+1, 3, h)}$
is equal to $3$ if and only if $\gcd(2h+1,q+1)>1$.
Equivalently, $d\ge 4$ if and only if $\gcd(2h+1,q+1)=1$.
\end{lem}

\begin{lem}\label{lem3.1}
Let $q$ be a positive power of a prime $p$.
Let $x_i\in U_{q+1}\subseteq \mathbb{F}_{q^{2}}$, $i\in\langle t\rangle$, be
arbitrary $t$ pairwise distinct elements and $r,s\in\mathbb{Z^{+}}$. Then the dimension of the solution space of
the system of homogeneous linear equations
\begin{align}\label{eq1}
\left(\begin{array}{cccc}
 x_{1}^{r} & x_{2}^{r} & \cdots & x_{t}^{r} \\
 x_{1}^{s} & x_{2}^{s} & \cdots & x_{t}^{s} \\
  \end{array}
\right){\bf X} = 0
\end{align}
over ${\mathbb{F}_{q}}$
is equal to the dimension of the solution space of the system of homogeneous linear equations
\begin{align}\label{eq2}
\left(\begin{array}{cccc}
 x_{1}^{r} & x_{2}^{r} & \cdots & x_{t}^{r} \\
 x_{1}^{s} & x_{2}^{s} & \cdots & x_{t}^{s} \\
 x_{1}^{-s} & x_{2}^{-s} & \cdots & x_{t}^{-s} \\
 x_{1}^{-r} & x_{2}^{-r} & \cdots & x_{t}^{-r} \\
  \end{array}
\right){\bf Y} = 0
\end{align}
over ${\mathbb{F}_{q^{2}}}$.
\end{lem}

\begin{proof}
Let $\beta$ be a generator of $U_{q+1}$.
Since $x_1,\cdots,x_t$ are distinct elements of $U_{q+1}$, we write $x_i=\beta^{l_i}$
for $i\in\langle t\rangle$, where $0\le l_1< \cdots< l_t\le q$.
Define $$S_1:=\left\{{\bf X}\in\mathbb{F}_q^t\mid A{\bf X}={\bf 0}\right\},d_1:=\dim_{\mathbb{F}_q}(S_1),$$
$$S_2:=\left\{{\bf Y}\in \mathbb{F}_{q^2}^t\mid B{\bf Y}={\bf 0}\right\},d_2:=\dim_{\mathbb{F}_{q^2}}(S_2),$$
where $A$ is the coefficient matrix of (\ref{eq1}) and $B$ is
the coefficient matrix of (\ref{eq2}).

Consider the linear isomorphism over $\mathbb{F}_q$:
\begin{align*}
\phi:\mathbb{F}_q^t & \to V_1,\\
(a_0,a_1,\cdots,a_{t-1})& \mapsto a_0x^{l_1}+a_1x^{l_2}+\cdots+a_{t-1}x^{l_t},
\end{align*}
where $V_1:=\mathrm{span}(x^{l_1},x^{l_2},\cdots,x^{l_t})\subseteq\mathbb{F}_q[x]$.
Under this map, $S_1$ corresponds to the set of polynomials $f(x)\in V_1$ such that
$f(\beta^r)=f(\beta^s)=0$. Let $g_1(x)$ and $g_2(x)$ be the minimal polynomials of
$\beta^r$ and $\beta^s$ over $\mathbb{F}_q$, respectively. Then $\mathrm{lcm}(g_1(x),g_2(x))|f(x)$.
Hence, $$S_1\cong \left\{ f(x)\in V_1\mid \mathrm{lcm}(g_1(x),g_2(x))|f(x)\right\}.$$

Now we consider the linear isomorphism over $\mathbb{F}_{q^2}$:
\begin{align*}
\psi:\mathbb{F}_{q^2}^t& \to V_2,\\
(b_0,b_1,\cdots,a_{t-1})& \mapsto b_0x^{l_1}+b_1x^{l_2}+\cdots+b_{t-1}x^{l_t},
\end{align*}
where $V_2:=\mathrm{span}(x^{l_1},x^{l_2},\cdots,x^{l_t})\subseteq\mathbb{F}_{q^2}[x]$.
Under this map, $S_2$ corresponds to the set of polynomials $f(x)\in V_2$ such that
$$f(\beta^r)=f(\beta^{-r})=f(\beta^s)=f(\beta^{-s})=0.$$
Since $\beta^{-r}=\beta^{rq}$ and $\beta^{-s}=\beta^{sq}$ are also roots of $g_1(x)$ and $g_2(x)$, respectively,
one has $$S_2\cong\left\{ f(x)\in V_2\mid\mathrm{lcm}(g_1(x),g_2(x))| f(x)\right\}.$$
Therefore, we conclude that $d_1=d_2$.

This finishes the proof of Lemma \ref{lem3.1}.
\end{proof}

\begin{lem}\label{lem3.2}
Let $q$ be a positive power of a prime $p$ and let $x,y,z,w\in U_{q+1}\subseteq \mathbb{F}_{q^{2}}$
be four pairwise distinct elements. Then for all integer $h$, it holds in $\mathbb{F}_{q^{2}}$ that
\begin{equation}\label{eq5}
\det\left(\begin{array}{cccc}
        x^h & y^h & z^h & w^h \\
        x^{h+1} & y^{h+1} & z^{h+1} & w^{h+1} \\
        x^{-(h+1)} & y^{-(h+1)} & z^{-(h+1)} & w^{-(h+1)} \\
        x^{-h} & y^{-h} & z^{-h} & w^{-h}
      \end{array}
\right)=0
\end{equation} if and only if it holds in $\mathbb{F}_{q^{2}}$ that
\begin{equation}\label{eq6}
\det\left( \begin{array}{cc}
\frac{X_1^{2h+1}-1}{X_1-1} & \frac{X_2^{2h+1}-1}{X_2-1} \\
\frac{Y_1^{2h+1}-1}{Y_1-1} & \frac{Y_2^{2h+1}-1}{Y_2-1}
\end{array}
\right)=0,
\end{equation}
where $X_1=\frac{x}{w},X_2=\frac{y}{w},Y_1=\frac{x}{z}$
and $Y_2=\frac{X_2Y_1}{X_1}=\frac{y}{z}$.
\end{lem}

\begin{proof}
Let
$$H(x,y,z,w)=\left(\begin{array}{cccc}
        x^h & y^h & z^h & w^h \\
        x^{h+1} & y^{h+1} & z^{h+1} & w^{h+1} \\
        x^{-(h+1)} & y^{-(h+1)} & z^{-(h+1)} & w^{-(h+1)} \\
        x^{-h} & y^{-h} & z^{-h} & w^{-h}
      \end{array}
\right).$$
It follows that
\begin{align*}
\det(H(x,y,z,w))
=&\det\left(\begin{array}{cccc}
        (\frac{x}{w})^h & \frac{y}{w}^h & \frac{z}{w}^h & 1 \\
        (\frac{x}{w})^{h+1} & \frac{y}{w}^{h+1} & \frac{z}{w}^{h+1} & 1 \\
        (\frac{x}{w})^{-(h+1)} & \frac{y}{w}^{-(h+1)} & \frac{z}{w}^{-(h+1)} & 1 \\
        (\frac{x}{w})^{-h} & \frac{y}{w}^{-h} & \frac{z}{w}^{-h} & 1
      \end{array}
\right).\\
\end{align*}
Let $X=\frac{x}{w},Y=\frac{y}{w}$ and $Z=\frac{z}{w}$. Then we have
\begin{align*}
\det(H(x,y,z,w))
=& (XYZ)^{-h}\det\left(\begin{array}{cccc}
X^{2h} & Y^{2h} & Z^{2h} & 1\\
X^{2h+1} & Y^{2h+1} & Z^{2h+1} & 1\\
X^{-1} & Y^{-1} & Z^{-1} & 1 \\
1 & 1 & 1 & 1\\
\end{array}
\right)\\
=&(XYZ)^{-h}\det\left(\begin{array}{cccc}
X^{2h}-1 & Y^{2h}-1 & Z^{2h}-1 & 1\\
X^{2h+1}-1 & Y^{2h+1}-1 & Z^{2h+1}-1 & 1\\
X^{-1}-1 & Y^{-1}-1 & Z^{-1}-1 & 1 \\
0 & 0 & 0 & 1\\
\end{array}
\right)\\
=& (XYZ)^{-h}\det\left(\begin{array}{ccc}
X^{2h}-1 & Y^{2h}-1 & Z^{2h}-1\\
X^{2h+1}-1 & Y^{2h+1}-1 & Z^{2h+1}-1\\
X^{-1}-1 & Y^{-1}-1 & Z^{-1}-1
\end{array}
\right)\\
=&(XYZ)^{-h}(X-1)(Y-1)(Z-1)\det\small\left(\begin{array}{ccc}
\frac{X^{2h}-1}{X-1} & \frac{Y^{2h}-1}{Y-1} & \frac{Z^{2h}-1}{Z-1} \\
\frac{X^{2h+1}-1}{X-1} & \frac{Y^{2h+1}-1}{Y-1} & \frac{Z^{2h+1}-1}{Z-1} \\
X^{-1} & Y^{-1} & Z^{-1} \\
\end{array}
\right)\\
=&(XYZ)^{-(h+1)}(X-1)(Y-1)(Z-1) \\
&\times \det\left(\begin{array}{ccc}
\frac{X}{X-1}(X^{2h}-1) & \frac{Y}{Y-1}(Y^{2h}-1) & \frac{Z}{Z-1}(Z^{2h}-1) \\
\frac{X}{X-1}(X^{2h+1}-1) & \frac{Y}{Y-1}(Y^{2h+1}-1) & \frac{Z}{Z-1}(Z^{2h+1}-1) \\
1 & 1 & 1 \\
\end{array}
\right)\\
=&(XYZ)^{-(h+1)}(X-1)(Y-1)(Z-1) \\
&\times \det\left(\begin{array}{ccc}
\frac{X}{X-1}(X^{2h}-1) & \frac{Y}{Y-1}(Y^{2h}-1) & \frac{Z}{Z-1}(Z^{2h}-1) \\
\frac{X}{X-1}(X^{2h+1}-X^{2h}) & \frac{Y}{Y-1}(Y^{2h+1}-Y^{2h}) & \frac{Z}{Z-1}(Z^{2h+1}-Z^{2h}) \\
1 & 1 & 1 \\
\end{array}
\right)\\
=&(XYZ)^{-(h+1)}(X-1)(Y-1)(Z-1) \\
&\times \det\left(\begin{array}{ccc}
\frac{X}{X-1}(X^{2h}-1) & \frac{Y}{Y-1}(Y^{2h}-1) & \frac{Z}{Z-1}(Z^{2h}-1) \\
X^{2h+1} & Y^{2h+1} & Z^{2h+1} \\
1 & 1 & 1 \\
\end{array}
\right)\\
=&(XYZ)^{-(h+1)}(X-1)(Y-1)(Z-1) \\
&\times \det\left(\begin{array}{ccc}
\frac{X^{2h+1}-X}{X-1}-\frac{Z^{2h+1}-Z}{Z-1} & \frac{Y^{2h+1}-Y}{Y-1}-\frac{Z^{2h+1}-Z}{Z-1} & \frac{Z^{2h+1}-Z}{Z-1} \\
X^{2h+1}-Z^{2h+1} & Y^{2h+1}-Z^{2h+1} & Z^{2h+1} \\
0 & 0 & 1 \\
\end{array}
\right)\\
=&(XYZ)^{-(h+1)}(X-1)(Y-1)(Z-1) \\
&\times \det\left(\begin{array}{cc}
\frac{X^{2h+1}-X}{X-1}-\frac{Z^{2h+1}-Z}{Z-1} & \frac{Y^{2h+1}-Y}{Y-1}-\frac{Z^{2h+1}-Z}{Z-1} \\
X^{2h+1}-Z^{2h+1} & Y^{2h+1}-Z^{2h+1}
\end{array}
\right)\\
=&(XYZ)^{-(h+1)}(X-1)(Y-1)(Z-1) \\
&\times \det\left(\begin{array}{cc}
\frac{X^{2h+1}-1}{X-1}-\frac{Z^{2h+1}-1}{Z-1} & \frac{Y^{2h+1}-1}{Y-1}-\frac{Z^{2h+1}-1}{Z-1} \\
X^{2h+1}-Z^{2h+1} & Y^{2h+1}-Z^{2h+1}
\end{array}
\right)\\
=&(XYZ)^{-(h+1)}(X-1)(Y-1) \\
&\times \det\left(\begin{array}{cc}
\frac{Z-1}{X-1}(X^{2h+1}-1)-(Z^{2h+1}-1) & \frac{Z-1}{Y-1}(Y^{2h+1}-1)-(Z^{2h+1}-1) \\
X^{2h+1}-Z^{2h+1} & Y^{2h+1}-Z^{2h+1}
\end{array}
\right)\\
=&(XYZ)^{-(h+1)}(X-1)(Y-1) \\
&\times \det\left(\begin{array}{cc}
\frac{Z-X}{X-1}(X^{2h+1}-1) & \frac{Z-Y}{Y-1}(Y^{2h+1}-1) \\
X^{2h+1}-Z^{2h+1} & Y^{2h+1}-Z^{2h+1}
\end{array}
\right)\\
=&-(XY)^{-(h+1)}Z^{h-1}(X-1)(Y-1)(X-Z)(Y-Z) \\
&\times \det\left(\begin{array}{cc}
\frac{X^{2h+1}-1}{X-1} & \frac{Y^{2h+1}-1}{Y-1} \\
\frac{(\frac{X}{Z})^{2h+1}-1}{(\frac{X}{Z})-1} & \frac{(\frac{Y}{Z})^{2h+1}-1}{(\frac{Y}{Z})-1}
\end{array}
\right).
\end{align*}
Let $X_1=X=\frac{x}{w},X_2=Y=\frac{y}{w},Y_1=\frac{X}{Z}=\frac{x}{z}$ and $Y_2=\frac{Y}{Z}=\frac{y}{z}$.
Since $X,Y,Z\in U_{q+1}\backslash\{1:=1_{\mathbb{F}_{q^{2}}}\}\subseteq \mathbb{F}_{q^{2}}$ are
three pairwise distinct elements, one knows that
$$-(XY)^{-(h+1)}Z^{h-1}(X-1)(Y-1)(X-Z)(Y-Z)\ne 0.$$
It then follows immediately that (\ref{eq5}) holds in $\mathbb{F}_{q^{2}}$
if and only if (\ref{eq6}) holds in $\mathbb{F}_{q^{2}}$ as expected.

The proof of Lemma \ref{lem3.2} is complete.
\end{proof}

\begin{lem}\label{lem3.3}
Let $q=p^m$ with $p$ a prime and $m$ a positive integer.
Suppose $2h+1 \equiv p^{m'} \pmod{q+1}$, where $m'$
is a positive integer. Let $x,y,z,w \in U_{q+1}$ be four
pairwise distinct elements. Then the following are equivalent:

{\rm (1).} The quadruple $(x,y,z,w)$ satisfies the determinant
condition (\ref{eq5}).

{\rm (2).} We have
$$w\in\{\mathrm{\pi}_{a,b,c,d}(\xi)\mid \xi\in
\mathbb{F}_{p^t}^{\ast}\setminus\{1\}\},$$
where $t=\gcd(m,m')$, $\mathrm{\pi}_{a,b,c,d}:
\mathbb{F}_q\cup\{\infty\}\to U_{q+1},
\mathrm{\pi}_{a,b,c,d}(\xi)=\frac{a\xi + c}{b\xi + d}$,
$a=y(x-z)$, $b=x-z$, $c=x(z-y)$ and $d=z-y$.

{\rm (3).} The quadruple $(x,y,z,w)$ is the image under
the projective linear transformation $\mathrm{\pi}_{a,b,c,d}$
of $(0,\infty,1,\xi)$ for some $\xi \in \mathbb{F}_{p^t}^* \setminus \{1\}$.

\end{lem}

\begin{proof}
Let $x,y,z,w$ be four distinct elements in $U_{q+1}$. By Lemma \ref{lem3.2},
the determinant condition $\det(H(x,y,z,w)) = 0$ is equivalent to
\begin{align*}
\det\left( \begin{array}{cc}
\frac{X_1^{2h+1}-1}{X_1-1} & \frac{X_2^{2h+1}-1}{X_2-1} \\
\frac{Y_1^{2h+1}-1}{Y_1-1} & \frac{Y_2^{2h+1}-1}{Y_2-1}
\end{array}
\right)=0,
\end{align*}
where $X_1=\frac{x}{w},X_2=\frac{y}{w},Y_1=\frac{x}{z},
Y_2=\frac{X_2Y_1}{X_1}=\frac{y}{z}$.
Since $2h+1\equiv p^{m'}\pmod{q+1}$, one has $$\frac{\beta^{2h+1}-1}{\beta-1}=\frac{(\beta-1)^{p^{m'}}}
{\beta-1}=(\beta-1)^{p^{m'}-1}$$
for any element $\beta$ in $U_{q+1}$.
Substituting this into the determinant condition yields $\det(H(x,y,z,w))=0$
if and only if
$$\Big(\frac{(X_1-1)(Y_2-1)}{(X_2-1)(Y_1-1)}\Big)^{p^{m'}-1}=
\Big(\frac{(x-w)(y-z)}{(y-w)(x-z)}\Big)^{p^{m'}-1}=1$$
since $X_1-1,X_2-1,Y_1-1$ and $Y_2-1$ are all nonzero.

First of all, we prove that $(1)\Rightarrow(2)$.
Let $\xi:=\frac{(X_1-1)(Y_2-1)}{(X_2-1)(Y_1-1)}$.
Since $X_1,X_2,Y_1,Y_2\in U_{q+1}\setminus\{1\}
\subseteq\mathbb{F}_{q^2}$ and $Y_2=\frac{X_2Y_1}{X_1}$,
one has
\begin{align*}
\xi^{q}=& \frac{(X_1-1)^{q}(Y_2-1)^{q}}{(X_2-1)^{q}(Y_1-1)^{q}}
= \frac{(X_1^{-1}-1)(Y_2^{-1}-1)}{(X_2^{-1}-1)(Y_1^{-1}-1)}\\
=& \frac{(X_1Y_2)^{-1}(1-X_1)(1-Y_2)}{(X_2Y_1)^{-1}(1-X_2)(1-Y_1)}
= \frac{(X_1-1)(Y_2-1)}{(X_2-1)(Y_1-1)}=\xi.
\end{align*}
So $\xi\in \mathbb{F}_{q}$.
Let $t=\gcd(m,m')$. Since $\xi^{p^{m'}-1}=1$ and $\xi^{q-1}=1$,
we derive that
$$\xi^{\gcd(q-1,p^{m'}-1)}=\xi^{p^t-1}=1,$$
which implies that $\xi\in \mathbb{F}_{q^t}^{\ast}$.

For fixed distinct elements $x,y,z \in U_{q+1}$, we determine all
$w \in \mathbb{F}_{q^2}$ satisfying the determinant condition (\ref{eq5}).
We claim that $z-y+\xi(x-z)\ne 0$ for any element $\xi\in \mathbb{F}_{q^t}^{\ast}$.
Suppose that there is an element $\xi$ such that $z-y+\xi(x-z)=0$.
Then
$$0=(z-y+\xi(x-z))^q=z^{-1}-y^{-1}+x^{-1}-z^{-1}=x^{-1}-y^{-1},$$
which contradicts with the hypothesis $x\ne y$. So the claim is proved.
It then follows from the claim and
$$\xi=\frac{(X_1-1)(Y_2-1)}{(X_2-1)(Y_1-1)}=\frac{(x-w)(y-z)}{(y-w)(x-z)}$$
that
$$w=\frac{x(z-y)+\xi y(x-z)}{(z-y)+\xi(x-z)}.$$

Via the bijection in Lemma \ref{U_q+1_PG(1,q)}, we can identify
$U_{q+1}$ with the projective line $\mathrm{PG}(1,q)$.
Under this identification, the action of
$\mathrm{PGL}(2,q)$ on $U_{q+1}$ is naturally identified with its action on $\mathrm{PG}(1,q)$.
To express this in terms of a projective linear transformation, we define the map
\begin{align*}
\mathrm{\pi}_{a,b,c,d}:\mathbb{F}_q\cup\{\infty\}&\to\mathbb{F}_{q^2},\\
\xi &\mapsto\frac{a\xi + c}{b\xi + d},
\end{align*}
where $a=y(x-z),b=x-z,c=x(z-y)$ and $d=z-y$.
By Lemma \ref{PGL2}, we know that $\mathrm{\pi}_{a,b,c,d}$ is a
projective permutation on $U_{q+1}$, which yields that $\mathrm{\pi}_{a,b,c,d}(\xi)\in U_{q+1}$
for any $\xi\in\mathbb{F}_q\cup\{\infty\}$. A direct computation shows that
\begin{align*}
\mathrm{\pi}_{a,b,c,d}(0) &= \frac{c}{d} = \frac{x(z-y)}{z-y} = x, \\
\mathrm{\pi}_{a,b,c,d}(\infty) &= \frac{a}{b} = \frac{y(x-z)}{x-z} = y, \\
\mathrm{\pi}_{a,b,c,d}(1) &= \frac{a+c}{b+d} = \frac{y(x-z) + x(z-y)}{(x-z) + (z-y)} = \frac{xy - yz + xz - xy}{x - y} = z.
\end{align*}
Thus $\mathrm{\pi}_{a,b,c,d}(\xi)\in U_{q+1}\setminus\{x,y,z\}$ if and only if $\xi\in\mathbb{F}_q^{\ast}\setminus\{1\}$.

Therefore, if (1) holds, then there exists an element $\xi\in\mathbb{F}_q^{\ast}\setminus\{1\}$
such that $w=\mathrm{\pi}_{a,b,c,d}(\xi)$, that is
$w\in\{\mathrm{\pi}_{a,b,c,d}(\xi))\mid \xi\in\mathbb{F}_{p^t}^{\ast}\setminus\{1\}\}$.
Hence (2) follows.

Now, we show that $(2)\Rightarrow (3)$. It follows directly from
$\mathrm{\pi}_{a,b,c,d}(0) = x$, $\mathrm{\pi}_{a,b,c,d}(\infty) = y$, $\mathrm{\pi}_{a,b,c,d}(1) = z$
and $w\in\{\mathrm{\pi}_{a,b,c,d}(\xi))\mid \xi\in\mathbb{F}_{p^t}^{\ast}\setminus\{1\}\}$.

Finally, we prove that $(3)\Rightarrow (1)$. If $(x,y,z,w)=\mathrm{\pi}_{a,b,c,d}(0,\infty,1,\xi)$
  for some $\xi\in\mathbb{F}_{p^t}^{\ast}\setminus\{1\}$, then
  \begin{align*}
  \left(\frac{(x-w)(y-z)}{(y-w)(x-z)}\right)^{p^{m'}-1}
  =\xi^{p^{m'}-1}=1.
  \end{align*}
  Namely, $\det(H(x,y,z,w))=0$ as desired.

This completes the proof of Lemma \ref{lem3.3}.
\end{proof}

\section{Proof of Theorem \ref{thm1.1}}
In this section, we give the proof of Theorem \ref{thm1.1}.\\
\\
{\it Proof of Theorem \ref{thm1.1}.}
First, since $\mathcal{C}$ is an AMDS code with parameters $[n,k,n-k]$,
it follows that for any subset $I\subseteq\langle n\rangle$ with $|I|=n-k-1$,
we have $\dim(\mathrm{span}(\{H_i,i\in I\}))=n-k-1$,
where $H_i$ denotes the $i$-th column of
the parity-check matrix $H$ of $\mathcal{C}$.

On the other hand, since the minimum weight codewords of
$\mathcal{C}$ support an $(n-k-1)-(n,n-k,\lambda)$ design, denoted by $\mathbb{D}=(\mathbb{P},\mathbb{B})$,
where $\mathbb{P}=\langle n\rangle$ and $$\mathbb{B}=\{ I\cup\{j\} \mid
I\subseteq \langle n\rangle,|I|=n-k-1,
j\in\mathrm{span}(\{H_i,i\in I\})\setminus I \},$$
we conclude that $$\left|\{j\in\langle n\rangle\mid H_j\in\mathrm{span}(\{H_i,i\in I\})\}\right|=n-k-1+\lambda.$$
Applying Lemma \ref{lem_structure} yields that
$\mathcal{C}^{\bot}(\langle n\rangle\setminus I)=\mathrm{span}(\{\mathbf{c}^{\bot}\})$,
where $$\mathrm{supp}(\mathbf{c}^{\bot})=\langle n\rangle\setminus
\{j\in\langle n\rangle\mid H_j\in\mathrm{span}(\{H_i,i\in I\})\}.$$
Therefore, $$\left|\mathrm{supp}(\mathbf{c}^{\bot})\right|=n-(n-k-1+\lambda)=k+1-\lambda,$$
which implies that the minimum distance $d^{\bot}$ of $\mathcal{C}^{\bot}$ is $k+1-\lambda$.
By the arbitrariness of the subset $I$, there does not exist a codeword
$\mathbf{c}^{\bot}\in \mathcal{C}^{\bot}$ such that
$k+1-\lambda<w(\mathbf{c}^{\bot})\le k+1$.
Moreover, the support sets of all the minimum weight codewords in
$\mathcal{C}^{\bot}$ are $$\{\langle n\rangle\setminus
\{j\in\langle n\rangle\mid H_j\in\mathrm{span}(\{H_i,i\in I\})\}\mid
I\subseteq\langle n\rangle,\left|I\right|=n-k-1\}.$$ It then follows that
the block set of ${\mathbb{D}^{\bot}}'$ is
$${\mathbb{B}^{\bot}}'=\{\cup_{B\in\mathbb{B}\atop I\subseteq B}B\mid
I\subseteq\mathbb{P},\left|I\right|=n-k-1\}.$$
Consequently, ${\mathbb{D}^{\bot}}'$ is an $(n-k-1)$-$(n,n-k-1+\lambda,1)$
design uniquely determined by $\mathbb{D}$. Thus,
by Lemma \ref{cite_complementary_design}, we conclude that the minimum
weight codewords in $\mathcal{C}^{\bot}$ support an $(n-k-1)$-$(n,k+1-\lambda,\lambda^{\bot})$ design $\mathbb{D}^{\bot}$,
where $\lambda^{\bot}=\frac{\binom{k+1}{n-k-1+\lambda}}{\binom{k+1}{\lambda}}$.

Furthermore, for any subset $I\subseteq\langle n\rangle$ with $\left|I\right|= n-k-1$,
we know that $\mathcal{C}^{\bot}(\langle n\rangle\setminus I)=\mathrm{span}(\{\mathbf{c}^{\bot}\})$
for some nonzero codeword $\mathbf{c}^{\bot}\in \mathcal{C}^{\bot}$
satisfying $\left|\mathrm{supp}(\mathbf{c}^{\bot})\right|=k+1-\lambda$.
For any subset $S\subseteq\langle n\rangle$, we define
$$g(S):=\left|\{\mathbf{c}^{\bot}\in\mathcal{C}^{\bot}\mid
\mathrm{supp}(\mathbf{c}^{\bot})\subseteq S\}\right|
=\left|\mathcal{C}^{\bot}(S)\right|$$ and
$$f(S):=\left|\{\mathbf{c}^{\bot}\in\mathcal{C}^{\bot}\mid
\mathrm{supp}(\mathbf{c}^{\bot})= S\}\right|.$$
Then one has
$$g(S)=\left|\mathcal{C}^{\bot}(S)\right|
=\left|\mathcal{C}^{\bot}(\langle n\rangle
\setminus(\langle n\rangle\setminus S))\right|
=q^{n-k-\dim(\mathrm{span}(\{ H_j , j \notin S \}))}.$$
The values of $g(S)$ can be determined by considering the two cases $|S| \leq k$ and $|S| \geq k+1$.

Case 1. $|S| \ge k+1$. We have $g(S)=q^{|S|-k}$, which follows from the fact that
any $n-|S|$ columns of the parity-check matrix $H$ are linearly independent.

Case 2. $|S| \le k$. Pick a subset $I\subseteq\langle n\rangle$ such that $S\subseteq I$ and $|I|=k+1$.
By Lemma \ref{lem_structure}, $\mathcal{C}^{\bot}(I)=\mathcal{C}^{\bot}
(\langle n\rangle\setminus(\langle n\rangle\setminus I)) =\mathrm{span}(\{\mathbf{c}^{\bot}\})$,
where $\mathbf{c}^{\bot}$ is a nonzero codeword in $\mathcal{C}^{\bot}$ satisfying that
$\left|\mathrm{supp}(\mathbf{c}^{\bot})\right|=k+1-\lambda$.
We denote $I_{\mathbf{c}^{\bot}}:=\mathrm{supp}(\mathbf{c}^{\bot})$.
Since $\mathcal{C}^{\bot}(S)\subseteq\mathcal{C}^{\bot}(I)
=\mathrm{span}(\{\mathbf{c}^{\bot}\})$,
we have $g(S)=q$ if and only if $I_{\mathbf{c}^{\bot}}\subseteq S$,
$g(S)=1$ if and only if $I_{\mathbf{c}^{\bot}}\not\subseteq S$.
This implies that there are no codewords in $\mathcal{C}^{\bot}$
whose weights in $\langle k+1\rangle \setminus \{k+1-\lambda\}$.

For any subset $S\subseteq\langle n\rangle$ with
$\left|S\right|=k+r$ and $r\in\langle n-k\rangle$,
by Lemma \ref{Mobius inversion formula}, we have
\begin{align*}
f(S)=&\sum_{T\subseteq S}(-1)^{|S|-|T|}g(T)\\
=&\sum_{i=0}^{k+r}\sum_{T\subseteq S\atop |T|=i}(-1)^{k+r-i}g(T)\\
=&\sum_{i=0}^{k-\lambda}\sum_{T\subseteq S\atop |T|=i}(-1)^{k+r-i}
+\sum_{i=k+1}^{k+r}\sum_{T\subseteq S\atop |T|=i}(-1)^{k+r-i}q^{i-k}\\
&+\sum_{i=k+1-\lambda}^{k}\sum_{T\subseteq S\atop |T|=i}(-1)^{k+r-i}g(T)
\sum_{T\subseteq I\subseteq S\atop |I|=k+1}\frac{1}{\binom{k+r-i}{k+1-i}}\\
=&\sum_{i=0}^{k-\lambda}(-1)^{k+r-i}\binom{k+r}{i}
+\sum_{i=k+1}^{k+r}\binom{k+r}{k+r-i}(-1)^{k+r-i}q^{i-k}\\
&+\sum_{i=k+1-\lambda}^{k}\frac{(-1)^{k+r-i}}{\binom{k+r-i}{k+1-i}}
\sum_{I\subseteq S\atop |I|=k+1}\sum_{T\subseteq I\atop |T|=i}g(T)\\
=&\sum_{i=0}^{k-\lambda}(-1)^{k+r-i}\binom{k+r}{i}
+\sum_{j=0}^{r-1}\binom{k+r}{j}(-1)^{j}q^{r-j}\\
&+\sum_{i=k+1-\lambda}^{k}\frac{(-1)^{k+r-i}}{\binom{k+r-i}{k+1-i}}
\sum_{I\subseteq S\atop |I|=k+1}\left(\sum_{T\subseteq I, |T|=i\atop I_{\mathbf{c}^{\bot}}\subseteq T}q
+\sum_{T\subseteq I, |T|=i\atop I_{\mathbf{c}^{\bot}}\not\subseteq T}1\right)\\
=&\sum_{i=0}^{k-\lambda}(-1)^{k+r-i}\binom{k+r}{i}
+\sum_{j=0}^{r-1}\binom{k+r}{j}(-1)^{j}q^{r-j}
+\sum_{i=k+1-\lambda}^{k}\frac{(-1)^{k+r-i}}{\binom{k+r-i}{k+1-i}}\\
&\sum_{I\subseteq S\atop |I|=k+1}\left(q\binom{k+1-(k+1-\lambda)}{i-(k+1-\lambda)}
+\binom{k+1}{i}-\binom{k+1-(k+1-\lambda)}{i-(k+1-\lambda)}\right)\\
=&\sum_{i=0}^{k-\lambda}(-1)^{k+r-i}\binom{k+r}{i}
+\sum_{j=0}^{r-1}\binom{k+r}{j}(-1)^{j}q^{r-j}\\
&+\sum_{i=k+1-\lambda}^{k}\frac{(-1)^{k+r-i}}{\binom{k+r-i}{k+1-i}}
\binom{k+r}{k+1}\left((q-1)\binom{\lambda}{k+1-i}+\binom{k+1}{i}\right)\\
=&\sum_{i=0}^{k-\lambda}(-1)^{k+r-i}\binom{k+r}{i}
+\sum_{j=0}^{r-1}\binom{k+r}{j}(-1)^{j}q^{r-j}\\
&+(q-1)\binom{k+r}{k+1}\sum_{i=k+1-\lambda}^{k}\frac{(-1)^{k+r-i}
\binom{\lambda}{k+1-i}}{\binom{k+r-i}{k+1-i}}
+\sum_{i=k+1-\lambda}^{k}\frac{(-1)^{k+r-i}\binom{k+r}{i}
\binom{k+r-i}{k+1-i}}{\binom{k+r-i}{k+1-i}}\\
=&\sum_{i=0}^{k-\lambda}(-1)^{k+r-i}\binom{k+r}{i}
+\sum_{j=0}^{r-1}\binom{k+r}{j}(-1)^{j}q^{r-j}\\
&+(q-1)\binom{k+r}{k+1}\sum_{i=k+1-\lambda}^{k}
\frac{(-1)^{k+r-i}\binom{\lambda}{k+1-i}}{\binom{k+r-i}{k+1-i}}
+\sum_{i=k+1-\lambda}^{k}(-1)^{k+r-i}\binom{k+r}{i}\\
=&\sum_{i=0}^{k}(-1)^{k+r-i}\binom{k+r}{i}
+\sum_{j=0}^{r-1}\binom{k+r}{j}(-1)^{j}q^{r-j}\\
&+(q-1)\binom{k+r}{k+1}\sum_{i=k+1-\lambda}^{k}\frac{(-1)^{k+r-i}
\binom{\lambda}{k+1-i}}{\binom{k+r-i}{k+1-i}}\\
=&(-1)^r\binom{k+r-1}{k}+\sum_{j=0}^{r-1}\binom{k+r}{j}(-1)^{j}q^{r-j}\\
&+(q-1)\binom{k+r}{k+1}\sum_{i=k+1-\lambda}^{k}\frac{(-1)^{k+r-i}
\binom{\lambda}{k+1-i}}{\binom{k+r-i}{k+1-i}}.
\end{align*}

Now we calculate the number of codewords in $\mathcal{C}^{\bot}(S)$
of Hamming weight $k+1-\lambda$.
Define
$$\mathbb{B}_S:=\{S\setminus \mathrm{supp}(\mathbf{c}^{\bot})
\mid \mathbf{c}^{\bot}\in\mathcal{C}^{\bot}(S),w(\mathbf{c}^{\bot})=k+1-\lambda\}.$$
Since ${\mathbb{D}^{\bot}}'$ is an $(n-k-1)$-$(n,n-k-1+\lambda,1)$ design,
it follows that
$\mathbb{D}_S=(S,\mathbb{B}_S)$ is an $(|S|-k-1)$-$(|S|,|S|-k-1+\lambda,1)$
design obtained by deleting the points outside $S$.
Lemma \ref{cite_complementary_design} gives us that the number
of codewords in $\mathcal{C}^{\bot}(S)$ of Hamming weight
$k+1-\lambda$ is exactly $$(q-1)\frac{\binom{k+1}{|S|-k-1+\lambda}}{\binom{k+1}{\lambda}}.$$
Therefore, the weight enumerator of $\mathcal{C}^{\bot}(S)$
with $|S|=k+r$ and $r\in\langle n-k\rangle$ is given by
\begin{align*}
&A_{\mathcal{C}^{\bot}(S)}(z)=\sum_{T\subseteq S}f(S)z^{|S|}
=\sum_{i=0}^{k+r}\sum_{T\subseteq S\atop |T|=i}f(T)z^{i}
=\sum_{i=0}^{k+1}\sum_{T\subseteq S\atop |T|=i}f(T)z^{i}
+\sum_{i=k+2}^{k+r}\sum_{T\subseteq S\atop |T|=i}f(T)z^{i}\\
=&1+(q-1)\prod_{j=1}^{r-1}\frac{k+1+j}{\lambda+j}z^{k+1-\lambda}
+\sum_{i=k+2}^{k+r}z^{i}\sum_{T\subseteq S\atop |T|=i}(-1)^{i-k}\binom{i-1}{k}\\
&+\sum_{j=0}^{i-k-1}\binom{i}{j}(-1)^{j}q^{i-k-j}
+(q-1)\binom{i}{k+1}\sum_{j=k+1-\lambda}^{k}\frac{(-1)^{i-j}
\binom{\lambda}{k+1-j}}{\binom{i-j}{k+1-j}}\\
=&1+(q-1)\prod_{j=1}^{r-1}\frac{k+1+j}{\lambda+j}z^{k+1-\lambda}
+\sum_{i=k+2}^{k+r}z^{i}\binom{k+r}{i}\Big((-1)^{i-k}\binom{i-1}{k}\\
&+\sum_{j=0}^{i-k-1}\binom{i}{j}(-1)^{j}q^{i-k-j}
+(q-1)\binom{i}{k+1}\sum_{j=k+1-\lambda}^{k}\frac{(-1)^{i-j}
\binom{\lambda}{k+1-j}}{\binom{i-j}{k+1-j}}\Big).
\end{align*}

This finishes the proof of Theorem \ref{thm1.1}. \hfill$\Box$

\section{Proofs of Theorems \ref{thm1.2} and \ref{thm1.3}}

In this section, we show Theorems \ref{thm1.2} and \ref{thm1.3}.
Firstly, we give the proof of Theorem \ref{thm1.2}.\\
\\
{\it Proof of Theorem \ref{thm1.2}.}
Applying Lemma \ref{XWLC-IT25} to the case
$2h+1\equiv p^{m'}\pmod{q+1}$, one derives that the minimum distance
$d$ of the BCH code $\mathcal{C}_{(q, q+1, 3, h)}$ satisfies $d\ge 4$.
Let $\mathcal{C}$ denote the BCH code $\mathcal{C}_{(q, q+1, 3, h)}$
and pick an element $\beta\in U_{q+1}$ being a generator.
Let $g_1(x)$ and $g_2(x)$ be the minimal polynomials
of the elements $\beta^h$ and $\beta^{h+1}$, respectively.
Through the finite field conjugate analysis, we observe
that $g_1(x)$ completely splits into factors containing
the roots $\beta^{\pm h}$ and $g_2(x)$ similarly factors
with roots $\beta^{\pm(h+1)}$. Noting that
$p^m\equiv-1\pmod{q+1}$ and $p^{2m}\equiv 1\pmod{q+1}$,
we derive that the order of $p$ modulo $q+1$ is $2m$.
Therefore, $p^{m'}\equiv 1\pmod{q+1}$ if and only if
$m'\equiv 0\pmod{2m}$, $p^{m'}\equiv -1\pmod{q+1}$
if and only if $m'\equiv m\pmod{2m}$.

(i).
If $m'\equiv 0\pmod{m}$, then exactly one of $2h$
or $2(h+1)$ is congruent to $0$ modulo $q+1$.
Thus $\deg(g_1(x)g_2(x))=\deg(g_1(x))+\deg(g_2(x))=3$.
This implies that the dimension of $\mathcal{C}$ is $q+1-3=q-2$.
Consequently, $\mathcal{C}$ is an MDS code with parameters $[q+1,q-2,4]$.
So (i) is proved.

(ii).
If $m'\not\equiv 0\pmod{m}$, then both of $2h$ and
$2(h+1)$ are not congruent to $0$ modulo $q+1$.
Hence $\deg(g_1(x)g_2(x))=\deg(g_1(x))+\deg(g_2(x))=4$.
This implies that the dimension of $\mathcal{C}$ is $q+1-4=q-3$.
Since $\mathcal{C}$ is actually the solution space of
\begin{align*}
\left(\begin{array}{cccc}
\beta^h & (\beta^2)^h & \cdots & (\beta^{q+1})^h \\
\beta^{h+1} & (\beta^2)^{h+1} & \cdots & (\beta^{q+1})^{h+1} \\
\end{array}
\right){\bf X}=0
\end{align*}
over ${\mathbb{F}_q}$, by Lemma \ref{lem3.1}, \ref{lem3.2}
and the theory of linear algebra over finite fields, every
codeword with Hamming weight $4$ in $\mathcal{C}$, up to
multiple, must correspond to a pairwise distinct quadruple
$(x, y, z, w)$ with $x,y,z,w\in U_{q+1}$ satisfies the
determinant condition (\ref{eq5}).
Since $p^t=2$, we have $\mathbb{F}_{q^t}^{\ast}\setminus\{1\}=\emptyset$.
Then by Lemma \ref{lem3.3}, one can deduce that there is no quadruple
$(x,y,z,w)$ satisfying (\ref{eq5}), which yields that $d>4$.
Meanwhile, Lemma \ref{singleton} tells us that $d\le 5$.
Consequently, $d=5$ and $\mathcal{C}$ is an MDS code with parameters $[q+1,q-3,5]$.
Thus (ii) is proved.

(iii).
Similar to the proof of (ii), one obtains the dimension of $\mathcal{C}$ is $q-3$.
Since $p^t>2$, one has $\mathbb{F}_{q^t}^{\ast}\setminus\{1\}\ne\emptyset$.
It then follows from Lemma \ref{lem3.1}, \ref{lem3.2} and \ref{lem3.3} that
$\mathcal{C}$ is an AMDS code with parameters $[q+1,q-3,4]$
and the minimum weight codeword in
$\mathcal{C}$ support a $3$-$(q+1,4,p^t-2)$ design $\mathbb{D}=(\mathcal{P},\mathcal{B})$,
where $\mathcal{P}=U_{q+1}$ and $$\mathcal{B}=\{\mathrm{\pi}_{a,b,c,d}(0,\infty,1,\xi)
\mid \xi\in\mathbb{F}_{p^t}^{\ast}\setminus\{1\},\{x,y,z\}\subseteq U_{q+1} {\rm\  pairwise\ distinct}\},$$
$a = y(x-z)$, $b = x-z$, $c = x(z-y)$, $d = z-y$, $\mathrm{\pi}_{a,b,c,d}(\xi)=\frac{a\xi + c}{b\xi + d}$
is a projective bijection between $\mathbb{F}_q\cup\{\infty\}$ and $U_{q+1}$.
Hence (iii) is proved.

This completes the proof of Theorem \ref{thm1.2}.\hfill$\Box$\\

Finally, we give the proof of Theorem \ref{thm1.3}.\\
\\
{\it Proof of Theorem \ref{thm1.3}.}
First of all, parts (i) and (ii) follow immediately from parts (i) and (ii)
of Theorem \ref{thm1.2} and Lemma \ref{lem_MDS}.

In the following let us show part (iii). Let $\mathcal{C}^{\bot}$ denote the dual code of BCH code
$\mathcal{C}_{(q, q+1, 3, h)}$ and let $\beta$ be a generator of $U_{q+1}$.
By Theorem \ref{thm1.2}, $\mathcal{C}$ is an AMDS code with parameters $[q+1,q-3,4]$
and the minimum weight codeword in $\mathcal{C}$ support a $3$-$(q+1,4,p^t-2)$ design.
Applying Theorem \ref{thm1.1}, we can conclude that $\mathcal{C}^{\bot}$
is a linear code with parameter $[q+1, 4, q-p^t]$
and the minimum weight codeword in $\mathcal{C}^{\bot}$
support a $3$-$(q+1,q-p^{t},\lambda^{\bot})$ design $\mathbb{D}^{\bot}=(\mathcal{P},\mathcal{B}^{\bot})$,
where $\lambda^{\bot}=\frac{\binom{q-2}{p^t+1}}{\binom{q-2}{p^t-2}}$ and
$\mathcal{B}^{\bot}=\{\langle n\rangle\setminus\cup_{B\in \mathbb{B}\atop I\subseteq B}B\mid I=\mathbb{P},|I|=3\}$.
Thus the block set ${\mathbb{B}^{\bot}}'$ of its complementary design
${\mathbb{D}^{\bot}}'=(\mathbb{P},{\mathbb{B}^{\bot}}')$ is equal to the set
$$\Big\{\bigcup_{B\in \mathbb{B}\atop I\subseteq B}B\Big| I\subseteq\mathbb{P},|I|=3\Big\}.$$
It follows from Theorem \ref{thm1.2} that
$${\mathbb{B}^{\bot}}'=\{\cup_{\xi\in \mathbb{F}_{p^t}\cup\{\infty\}}\mathrm{\pi}_{a,b,c,d}(\xi)
\mid \{x,y,z\}\subseteq U_{q+1}, x\ne y\ne z\ne x\},$$
where $a=y(x-z),b=x-z,c=x(z-y),d=z-y$,
$\mathrm{\pi}_{a,b,c,d}(\xi) = \frac{a\xi + c}{b\xi + d}$
is a projective bijection between $\mathbb{F}_q\cup\{\infty\}$
and $U_{q+1}$. Hence, the design $\mathbb{D}'$ is actually a Witt
spherical geometry design with parameters $3$-$(q+1,p^{\gcd(m,i)}+1,1)$.
Moreover, by Theorem \ref{thm1.1} and \ref{thm1.2}, the weight enumerator
of $\mathcal{C}_{(q, q+1, 3, h)}^{\bot}$ is given by
\begin{align*}
A_{\mathcal{C}_{(q, q+1, 3, h)}^{\bot}}(z)
=&1+(q-1)\frac{(q-1)q(q+1)}{(p^t-1)p^t(p^t+1)}z^{q-p^t}
+\sum_{i=q-1}^{q+1}z^{i}\binom{q+1}{i}\\
&\Big((-1)^{i-q+3}\binom{i-1}{q-3}+\sum_{j=0}^{i-q+2}
\binom{i}{j}(-1)^{j}q^{i-q+3-j}\\
+&(q-1)\binom{i}{q-2}\sum_{j=q-p^t}^{q-3}\frac{(-1)^{i-j}
\binom{p^t-2}{q-2-j}}{\binom{i-j}{q-2-j}}\Big).
\end{align*}

This concludes the proof of Theorem \ref{thm1.3}. \hfill$\Box$

\section{Final Remarks}
In this section, we give several applications of Lemma \ref{lem_structure} and Theorem \ref{thm1.1}.

(1).
It is worth noting that Lemma \ref{lem_structure} is a generalization of Proposition 14 in \cite{[FW-DCC97]}.

Let $\mathcal{C}$ is an NMDS code over $\mathbb{F}_q$ with parameters $[n,k,n-k]$.
For any minimum weight codeword $\mathbf{c}\in\mathcal{C}$ and subset
$I\subseteq \mathrm{supp}(\mathbf{c})$ with size $n-k-1$,
we apply Lemma \ref{lem_structure} with $r=1$. Then we know that
$|I_0|=|\mathrm{supp}(\mathbf{c})|=n-k$ and $|I_1|=k$,
which follows that there exists, up to a multiple,
a unique minimum weight codeword $\mathbf{c}^{\bot}\in\mathcal{C}^{\bot}$
such that $\mathrm{supp}(\mathbf{c})\cap\mathrm{supp}(\mathbf{c}^{\bot})=\emptyset$.
This shows Lemma \ref{lem_structure} is a generalization of Proposition 14 in \cite{[FW-DCC97]}.

(2). One notes that Theorem 1.1 (ii) in \cite{[QWH-FFA25]} was used to confirm a conjecture
raised by Geng, Yang, Zhang and Zhou in \cite{[GYZZ-FFA22]}. Now we use Lemma \ref{lem_structure}
to give an alternative proof of Theorem 1.1 (ii) in \cite{[QWH-FFA25]}. The proof is given as follows.

At first, let $\mathcal{C}$ be an NMDS code with parameters $[n,k,n-k]$.
Then its dual code $\mathcal{C}^{\bot}$ is an NMDS code with parameters $[n,k^{\bot}=n-k,k]$.
For any subset $I\subseteq \langle n\rangle$ with size $n-k^{\bot}-1=k-1$,
by applying Lemma \ref{lem_structure} with $r=1$,
we get that $\dim(\mathcal{C}(\langle n\rangle\setminus I))=1$.
Hence, by the arbitrariness of subset $I$, we have $\dim(\mathcal{C}(S))=1$
for any subset $S\subseteq\langle n\rangle$ with $|S|=n-k+1$.

Conversely, suppose that $\mathcal{C}$ is an AMDS code with parameters $[n,k,n-k]$ and
$\dim(\mathcal{C}(S))=1$ holds for any subset $S\subseteq\langle n\rangle$ with $|S|=n-k+1$.
We claim that for any subset $I\subseteq \langle n\rangle$ with $|I|=n-k-1$, $|I_0|\le n-k$
($I_0$ follows the notation in Lemma \ref{lem_structure}). In fact, if we assume that there
is a subset $I\subseteq\langle n\rangle $ such that $|I_0|\ge n-k+1$,
then we can pick a subset $S\subseteq I_0$ such that $|S|=n-k+1$. It follows that
$$\dim(\mathcal{C}(S))=|S|-\mathrm{rank}_H(S)=n-k+1-(n-k-1)=2,$$
contradicting with the hypothesis $\dim(\mathcal{C}(S))=1$. Thus the claim is proved.
Then by Lemma \ref{lem_structure}, $|I_1|=n-|I_0|\ge k$,
so does $d(\mathcal{C}^{\bot})$ by the arbitrariness of subset $I$.
Since $\mathcal{C}^{\bot}$ is not an MDS code, it follows that $d(\mathcal{C}^{\bot})<k+1$.
Consequently, $d(\mathcal{C}^{\bot})=k$, which yields that $\mathcal{C}$ is an NMDS code as desired.

(3).
In 2021, Ding, Tang and Tonchev \cite{[DTT-DCC21]} used the general projective linear group
$\mathrm{PGL}(2,2^m)$ to study a class of AMDS codes $\mathcal{C}$
$\mathbb{F}_{2^m}$ with parameters $[2^m+1,2^m-3,4]$.
They proved that the supports of the codewords of Hamming weight $4$ in $\mathcal{C}$
form a $3$-$(2^m+1,4,2)$ design. However, they did not give the weight enumerators of $\mathcal{C}$
and its dual code $\mathcal{C}^{\bot}$. By Theorem \ref{thm1.1}, the weight enumerator
$A_{\mathcal{C}^{\bot}}(z)$ of the dual code $\mathcal{C}^{\bot}$ can be obtained immediately as:
\begin{align*}
A_{\mathcal{C}^{\bot}}(z)=& 1+(2^m-1)\prod_{j=1}^{3}\frac{2^m-2+j}{2+j}z^{2^m-4}
+\sum_{i=2^m-1}^{2^m+1}z^{i}\binom{2^m+1}{i}\\
&\biggl((-1)^{i+3-2^m}\binom{i-1}{2^m-3}+\sum_{j=0}^{i+2-2^m}\binom{i}{j}(-1)^{j}2^{m(i+3-2^m-j)}\\
&+(2^m-1)\binom{i}{2^m-2}\sum_{j=2^m-4}^{2^m-3}\frac{(-1)^{i-j}\binom{2}{2^m-2-j}}{\binom{i-j}{2^m-2-j}}\biggr).
\end{align*}
For $m=4$, we have
$$A_{\mathcal{C}^{\bot}}(z)= 1 + 1020 z^{12} + 24480 z^{15} + 15555 z^{16} + 24480 z^{17},$$
which coincides with Example 28 in \cite{[DTT-DCC21]}, thereby confirming our result.

(4).
In \cite{[XTL-DCC22]}, Xiang, Tang and Liu derived the weight enumerator via complex calculations.
However, Theorem \ref{thm1.1} yields that the weight enumerator $A_{\mathcal{C}^{\bot}}(z)$ of the dual code
can be obtained directly without resorting to such calculations:
\begin{align*}
A_{\mathcal{C}^{\bot}}(z)=& 1+(7^m-1)\prod_{j=1}^{3}\frac{7^m-2+j}{5+j}z^{7^m-7}
  +\sum_{i=7^m-1}^{7^m+1}z^{i}\binom{7^m+1}{i}\\
  &\biggl((-1)^{i+3-7^m}\binom{i-1}{7^m-3}+\sum_{j=0}^{i+2-7^m}\binom{i}{j}(-1)^{j}7^{m(i+3-7^m-j)}\\
  &+(7^m-1)\binom{i}{7^m-2}\sum_{j=7^m-7}^{7^m-3}\frac{(-1)^{i-j}\binom{5}{7^m-2-j}}{\binom{i-j}{7^m-2-j}}\biggr).
\end{align*}
For $m=2$, one has
$$A_{\mathcal{C}^{\bot}}(z)= 1 + 16800 z^{42} + 2469600 z^{48} + 808800 z^{49} + 2469600 z^{50}.$$
This agrees with Example 2 in \cite{[XTL-DCC22]}.

\bibliographystyle{amsplain}

\end{document}